\makeatletter
\newcommand{\linebreakand}{%
  \end{@IEEEauthorhalign}
  \hfill\mbox{}\par
  \mbox{}\hfill\begin{@IEEEauthorhalign}
}
\makeatother
\documentclass[conference]{IEEEtran}
\IEEEoverridecommandlockouts
\usepackage{cite}
\usepackage{amsmath,amssymb,amsfonts}
\usepackage{algorithmic}
\usepackage{graphicx}
\usepackage{textcomp}
\usepackage{xcolor}

\usepackage{url}            % 更好的网址显示
\usepackage{listings}       % 代码块
\usepackage{xurl}
\usepackage{makecell}       % 高级单元格
\usepackage{tabularx}       % 弹性宽度表格
\usepackage{threeparttable} % 表格注脚
\usepackage{array}          % 表格增强
\usepackage{booktabs}       % 三线表
\usepackage{textcomp}       % 提供“直引号”符号支持
\usepackage[T1]{fontenc}    % 使用 T1 字体编码，修正字符映射
\usepackage[varqu]{zi4}     % 加载 Inconsolata 字体

\usepackage[hidelinks,bookmarks=false]{hyperref}

\newtheorem{theorem}{Theorem}
\newtheorem{definition}{Definition}

\definecolor{jsoningray}{gray}{0.95}   % 背景浅灰色
\definecolor{jsonkey}{RGB}{128,0,128}  % Key颜色: 紫色
\definecolor{jsonval}{gray}{0.0}       % Value颜色: 黑色 (默认)
\definecolor{jsonred}{RGB}{200,0,0}    % 后缀颜色: 红色 (用于 @@包裹内容)
\definecolor{jsonnum}{RGB}{0,0,255}    % 数字颜色: 蓝色

\lstdefinelanguage{json}{
    basicstyle=\footnotesize\ttfamily\color{jsonval}, % 字体符合IEEE(小号等宽)，默认颜色设为黑色
    columns=flexible,       % 强制固定列宽，解决“看起来乱、间隙大”的问题
    upquote=true,        % 强制使用“直引号”，解决引号像中文/弯曲的问题
    backgroundcolor=\color{jsoningray}, % 背景浅灰
    frame=single,                       % 四周有边框
    framerule=0.5pt,                    % 边框粗细
    rulecolor=\color{black!50},         % 边框颜色(深灰)
    framesep=3pt,                       % 内容与边框的间距
    numbers=left,
    numbersep=5pt,
    numberstyle=\tiny\color{gray},
    stringstyle=\color{jsonval},        % 普通字符串(Value)显示为黑色
    keywordstyle=\color{jsonkey}\bfseries, % 关键字(Key)显示为紫色加粗
    alsoletter={"-},                    % 让双引号成为单词一部分，以便匹配 "key"
    moredelim=**[is][\color{jsonred}]{@@}{@@},
    literate=
     {0}{{{\color{jsonnum}0}}}{1}
     {1}{{{\color{jsonnum}1}}}{1}
     {2}{{{\color{jsonnum}2}}}{1}
     {3}{{{\color{jsonnum}3}}}{1}
     {4}{{{\color{jsonnum}4}}}{1}
     {5}{{{\color{jsonnum}5}}}{1}
     {6}{{{\color{jsonnum}6}}}{1}
     {7}{{{\color{jsonnum}7}}}{1}
     {8}{{{\color{jsonnum}8}}}{1}
     {9}{{{\color{jsonnum}9}}}{1},
    breaklines=true,
    captionpos=b,
    tabsize=2,
    showstringspaces=false
}

\begin{document}

\title{AgentDID: Trustless Identity Authentication for AI Agents \\
}
\author{
\IEEEauthorblockN{Minghui Xu}
\IEEEauthorblockA{
\textit{Quancheng Laboratory \& Shandong University}\\
mhxu@sdu.edu.cn}
\and
\IEEEauthorblockN{Xiaoyu Liu}
\IEEEauthorblockA{ 
\textit{Shandong University}\\
lxy\_dawn@mail.sdu.edu.cn}
\and
\IEEEauthorblockN{Yihao Guo\thanks{Corresponding author: Yihao Guo}}
\IEEEauthorblockA{ 
\textit{The Hong Kong Polytechnic University}\\
yihao.guo@polyu.edu.hk}
\linebreakand
\IEEEauthorblockN{Chunchi Liu}
\IEEEauthorblockA{\textit{Huawei Technologies Co., Ltd.}\\
liuchunchi@gwu.edu}
\and
\IEEEauthorblockN{Yue Zhang}
\IEEEauthorblockA{
\textit{Quancheng Laboratory \& Shandong University}\\
zyueinfosec@sdu.edu.cn}
\and
\IEEEauthorblockN{Xiuzhen Cheng}
\IEEEauthorblockA{
\textit{Shandong University}\\
xzcheng@sdu.edu.cn}
}

\maketitle

\begin{abstract}
AI agents are autonomous entities that can be instantiated on demand, migrate across platforms, and interact with other agents or services without continuous human supervision. In such environments, identity is critical for establishing reliable interaction semantics among agents that may lack prior trust relationships. However, existing identity and access management mechanisms are designed for human users or static machines, assuming centralized enrollment, persistent identifiers, and stable execution contexts. These assumptions do not hold for AI agents, whose identities are self-managed, short-lived, and tightly coupled with their execution state and capabilities. We study the problem of identity authentication and state verification for AI agents and identify three challenges: (1) supporting self-managed identities for autonomously created agents, (2) enabling authentication under large-scale, concurrent interactions, and (3) verifying agents' dynamic execution state, such as whether their context and capabilities remain valid at interaction time. To address these challenges, we present AgentDID, a decentralized framework for identity authentication and state verification. AgentDID leverages decentralized identifiers (DIDs) and verifiable credentials (VCs), enabling agents to manage their own identities and authenticate across systems without centralized control. To address the limitations of static credential-based approaches, AgentDID introduces a challenge-response mechanism that allows verifiers to validate an agent's execution conditions at interaction time. We implement AgentDID in compliance with W3C standards and evaluate it through throughput experiments with multiple concurrent agents. Results show that the system achieves scalable identity authentication and state verification, demonstrating its potential to support large populations of AI agents.

%To solve the above challenges, we present AgentDID, a trustless identity authentication for AI agents. AgentDID adopts decentralized identifiers and verifiable credentials to enable agents to manage identities and authenticate across systems without centralized control. To address the static nature of credential-based approaches, AgentDID introduces a challenge–-response mechanism that allows verifiers to check agents’ execution conditions, such as execution context and capability availability. 

%We implement AgentDID following W3C standards and conduct throughput experiments. The results show that the system supports concurrent identity authentication and state verification for large numbers of AI agents. 
\end{abstract}


\section{Introduction}
\label{sec:introduction}
With the broader availability of large-scale models, AI agents are becoming key components in intelligent applications~\cite{yao2024survey,deng2025ai,yan2025protecting}. In 2024, several organizations, including OpenAI, Google, and Anthropic, released agent systems capable of autonomous task execution. As a result, the use of AI agents has expanded significantly in both scope and application. According to recent data released by CrewAI, the platform has executed more than 60 billion AI agent runs in total, with over 100 thousand multi-agent execution groups processed each day~\cite{CrewAI}. These observations illustrate the large-scale and high-concurrency characteristics of emerging agent ecosystems. The number of agents is expected to continue increasing. According to TRAY.AI, by 2027, 86\% of enterprises are expected to deploy AI agents, and 42\% plan to develop more than 100 agent prototypes~\cite{trayai2024agents}.  
One key factor behind this growth is the shift away from exclusive control by large technology companies. Open-source models such as the Llama series, together with local inference frameworks like Ollama, have reduced development complexity. This change has enabled individual users to build and operate their own agents, leading to a more decentralized ecosystem. 

However, the rapid expansion of the AI agent ecosystem has revealed critical security gaps~\cite{deng2025ai}. Current agent interaction frameworks generally lack effective identity authentication mechanisms. As a result, risks such as identity spoofing, unauthorized access, and privilege misuse are becoming increasingly prominent. Traditional identity and access management (IAM) systems, including OAuth 2.0, SAML, and FIDO2, were originally designed for scenarios involving users, devices, and platforms~\cite{yang2016model,arias2019survey}. These systems focus primarily on human users or static machines and are based on assumptions such as long-term sessions and centralized control of permissions. In contrast, AI agents operate independently without direct human intervention, which makes it difficult to maintain continuous control through conventional methods. 
More fundamentally, the nature of AI agents challenges the core assumptions of existing IAM models. These agents are capable of performing tasks autonomously across environments, interacting at very large scales, and adapting to changing operational conditions in real time. %These characteristics collectively give rise to new requirements for identity authentication and dynamic state verification~\cite{xi2025rise,ferrag2025llm}, introducing the following additional challenges (\textbf{C1-C3}). 
% These characteristics collectively give rise to new requirements for identity authentication and dynamic state verification~\cite{xi2025rise,ferrag2025llm}, 
% \textit{a key challenge of which is the ability to manage and verify identities for large populations of AI agents operating in decentralized environments}.
% These characteristics introduce new requirements for identity authentication and dynamic state verification~\cite{xi2025rise,ferrag2025llm}, leading to the question of how the identities of millions of AI agents can be managed and verified in a decentralized environment. 
These characteristics lead us to consider the following question: 

\textit{How can the identities of millions of AI agents be managed and verified in a decentralized environment?}

Addressing this question introduces new requirements for identity authentication and dynamic state verification~\cite{xi2025rise,ferrag2025llm}. We examine this problem along the following three dimensions (\textbf{C1--C3}). 

\textbf{(C1) Self-Managed Identity.} AI agents can be autonomously generated and operate independently, challenging centralized identity management systems that rely on pre-registered identities and fixed identifier assignment. Classical studies in multi-agent systems have modeled agents as autonomous entities with unique identifiers \cite{wooldridge2009introduction}. AI agents are capable of self-generation and independent operation. As they move between entities and systems, their identities must remain usable and verifiable. However, as agents increasingly migrate across heterogeneous platforms and engage in complex collaborative tasks, recent research has noted that their identities are often bound to specific platforms or environments~\cite{ferrag2025llm}. This limits the ability to recognize and verify agents across systems, leading to risks such as identity confusion. 

IAM systems rely on centralized registration and fixed identifier assignment. These approaches are not well-suited to support autonomous identity generation and verifiable attribute claims. Furthermore, they cannot reliably associate the same agent with a consistent identity as it interacts across organizational or platform boundaries. Therefore, a primary challenge in decentralized agent ecosystems is to construct an identity model that does not rely on a single provider, supports cross-system verification, and enables agents to present verifiable claims about their attributes. 

\textbf{(C2) Challenges of Large-Scale Authentication.} One of the notable characteristics of AI agents is their deployment at large scale with high levels of concurrency, which places increased load on existing centralized authentication systems. 

In such settings, agents communicate simultaneously with other agents, APIs, and external services. These interactions often span multiple domains and involve complex chains of requests. As the number of agents increases, authentication traffic grows proportionally, placing substantial pressure on centralized authentication infrastructures. Research on agentic AI has shown that centralized components for trust and authentication can become performance bottlenecks, reducing system availability and responsiveness under high concurrency. Centralized architectures face constraints in throughput, queuing, and consistency under load. Single-point processing models may not sustain reliable performance at scale.  
Furthermore, agent coordination at scale requires frequent identity verification. Relying on centralized authentication paths adds latency and increases the risk of failure at critical points in the system. In large-scale environments, enabling direct authentication between agents without depending on central intermediaries becomes an essential design requirement.  

\textbf{(C3) Dynamic State Verification.}
Current identity and access control systems lack mechanisms to verify the dynamic state of AI agents during execution. This limitation is critical because agent behavior depends on internal factors such as context, memory, intent, and current capabilities. These conditions change continuously and directly influence whether an agent should be allowed to perform a given action. Recent research confirms that agent decisions are shaped by evolving internal state over the course of multi-step tasks \cite{park2023generative}. However, most existing models rely on static roles or fixed identity attributes and do not account for changes in operational state. Studies have shown that without state verification, systems face risks such as privilege misuse, unauthorized operations, and insufficient auditability \cite{yi2024jailbreak}. Verifying the operational state of agents during authorization remains a core challenge in dynamic, multi-agent environments. 

In multi-agent interaction settings, identity authentication alone is insufficient to ensure secure interactions, as it only validates a peer’s long-term identifier and does not constrain the agent’s execution state during interaction. In practice, an adversarial agent may pass identity authentication and then mislead other agents by misrepresenting its capabilities, thereby influencing task assignment or interaction decisions.  This risk has been documented in real-world systems and security analyses. For example, Trustwave SpiderLabs’ analysis~\cite{Abusing} of the Agent-to-Agent (A2A) protocol shows that capability advertisement mechanisms based on static Agent Cards can be abused, allowing a malicious agent to dominate task allocation without violating identity authentication. These findings illustrate that identity authentication alone does not capture an agent’s execution state.  

Consequently, agent systems operating in open environments require both identity authentication and state verification, in order to ensure that an agent’s execution state remains consistent with its declared identity and capabilities.  

Decentralized Identifiers (DIDs) and Verifiable Credentials (VCs) are identity primitives supported by blockchains~\cite{xu2024exploring}. They have been applied in Web3~\cite{guo2025xrwa} and cross-chain environments~\cite{guo2023cross,guo2024zkCross} to support identity management and authentication across systems. DIDs allow agents to create and control identifiers without depending on a single provider. VCs provide a structured format for issuing and verifying claims about an agent’s permissions or capabilities. These features enable identity management that is both decentralized and verifiable. 
Based on these properties, we adopt DIDs and VCs to support the authentication of AI agents. Each agent is assigned a DID, which serves as a stable reference across systems and domains. Its permissions and execution context are expressed as VCs. This allows agents to authenticate directly, without relying on centralized services. However, the claims defined in VCs are typically static and do not reflect the agent’s state during execution. Relying solely on VCs is therefore insufficient for dynamic state verification~\cite{south2025identity,zou2025blocka2a}. To support dynamic state verification as required in C3, we propose a challenge–response mechanism. This mechanism requires agents to present evidence of current execution conditions, such as context or workload, at the time of interaction. The system can then make access decisions based on these current conditions rather than fixed attributes. This approach supports decentralized authentication while adapting to the dynamic nature of agent execution. 

For the sake of convenience, we highlight our contributions as follows: 
\begin{enumerate}

\item We propose AgentDID, a decentralized framework for AI agent authentication and state verification. To the best of our knowledge, this is the first framework that addresses dynamic state verification for AI agents in a decentralized setting. 

\item We follow the W3C standard and define custom credential types to represent agent-specific attributes at a fine-grained level. These credentials enable agents to authenticate across systems and environments (addressing challenges C1 and C2). To verify the correctness of an agent’s dynamic state, we design a challenge–response mechanism that checks dynamic execution conditions such as context and workload. This mechanism directly supports the requirements of C3. 

\item To validate the performance of our scheme, we conduct experiments on Sepolia, an Ethereum test network that simulates real-world deployment environments. We evaluate AgentDID from multiple perspectives, such as credential size, issuance time, and on-chain gas cost. Results show that latency remains within the seconds range and gas consumption stays below one US dollar, confirming the practicality of the approach. 

\end{enumerate}

The rest of the paper is organized as follows.
We review related work in Section~\ref{sec:related_work}.
Section~\ref{sec:preliminaries} introduces the system model and preliminaries.
AgentDID is presented in Section~\ref{sec:AgentDID}.
Section~\ref{sec:security-analysis} analyzes the security of the proposed scheme.
Implementation and evaluation are discussed in Section~\ref{sec:evaluation}.
Section~\ref{sec:conclusion} concludes the paper.

\section{Related Work} \label{sec:related_work}
In this section, we review related work on identity management and authentication, including traditional IAM systems and decentralized authentication approaches. 

\noindent \textbf{Traditional IAM.} 
Traditional identity and access management (IAM) systems address authentication for users and services in enterprise and web environments. 
These systems are typically based on centralized or federated trust models, in which identity providers authenticate subjects and issue credentials or assertions consumed by relying parties. 

Early IAM mechanisms follow centralized designs. 
Kerberos provides ticket-based authentication using a trusted key distribution center~\cite{kohl1993kerberos}, while directory services such as LDAP and Active Directory maintain identity records and credentials in centralized repositories~\cite{zeilenga2006lightweight}. 
For cross-domain authentication, federated identity standards such as SAML~2.0 define protocols through which identity providers issue signed assertions that can be verified by multiple service providers~\cite{saml2}. 
In web and API-oriented settings, OAuth~2.0 specifies an authorization framework based on delegated access, and OpenID Connect (OIDC) extends OAuth~2.0 by defining an identity layer that supports user authentication via ID tokens~\cite{hardt2012oauth,sakimura2014openid}. 
These protocols are implemented in both commercial and open-source IAM systems for managing access to applications and services. 
In addition, FIDO2 and WebAuthn define public-key–based authentication mechanisms that can be integrated into IAM systems to support passwordless authentication~\cite{webauthn}. 

Traditional IAM systems primarily verify static identities and enforce access control policies based on authenticated subjects. 
They generally assume that identity subjects are users or long-lived services and do not explicitly model autonomous agent-to-agent interactions or verification of runtime execution state at interaction time.

\noindent \textbf{Decentralized Authentication.} 
Decentralized authentication is typically grounded in the concept of self-sovereign identity (SSI), whose core objective is to enable cross-domain identification and authentication without relying on centralized identity providers. Under this paradigm, identity subjects generate and control their own identifiers, and authentication is achieved through cryptographic mechanisms rather than centralized account systems.  The foundation of decentralized identity systems consists of DIDs and VCs~\cite{mazzocca2025survey}. The DID specification proposed by the W3C defines a mechanism for managing identifiers without requiring a central registration authority. 
Building on DIDs, the VC specification further defines issuer-signed structured statements that attest to identity attributes or qualifications of a subject, and supports the verifiable presentation (VP) and validation of such claims across different application contexts. 

Certain works have explored the initial integration of decentralized authentication mechanisms with AI agents.  
Hu~{\em et al.}~\cite{hu2025trustless} investigate the motivation and potential benefits of combining decentralized technologies with AI agents.  
Huang~{\em et al.}~\cite{huang2025novel} study the use of DIDs and VCs to enable fine-grained identity descriptions for AI agents.   
Garzon~{\em et al.}~\cite{garzon2025ai} experimentally evaluate a DID-based approach in a multi-agent environment.  
BlockA2A~\cite{zou2025blocka2a} is an agent authentication framework built on decentralized identity, which introduces a Defense Orchestration Engine to enhance the security of agent interactions. 

Existing work on decentralized identity for AI agents primarily uses self-managed identities in A2A authentication, where DIDs and VCs support identity generation and verification. 
Many proposals~\cite{hu2025trustless,huang2025novel,garzon2025ai,zou2025blocka2a} focus on protocol design or system architecture, while empirical evaluation under concurrent multi-agent interaction workloads is not reported, leaving their applicability to agent concurrency settings unspecified. 
In addition, these approaches~\cite{hu2025trustless,huang2025novel,garzon2025ai,zou2025blocka2a} often treat agent identity as a static property and do not explicitly incorporate dynamic runtime states during interactions, such as workload, context consistency, or capability availability. 
Motivated by this scope, we present AgentDID, which studies decentralized authentication together with state-aware verification; a comparison is summarized in Table~\ref{tab:compare}.  
\begin{table}[t]
\centering
\caption{Coverage of requirements for AI agent interactions.}
\label{tab:compare}
\begin{tabular*}{\columnwidth}{@{\extracolsep{\fill}}lccc}
\toprule
\textbf{Entity / Approach} 
& \textbf{SSI} 
& \textbf{CS} 
& \textbf{State} \\
\midrule
Agent Requirements 
& Required 
& Required 
& Required \\
\midrule
Traditional IAM 
& No 
& Limited
& No \\
\midrule
Hu~{\em et al.}~\cite{hu2025trustless} 
& Yes 
& N/A 
& No \\
Huang~{\em et al.}~\cite{huang2025novel} 
& Yes 
& Theoretical  
& No \\
Garzon~{\em et al.}~\cite{garzon2025ai} 
& Yes 
& Theoretical  
& No \\
Zou~{\em et al.}~\cite{zou2025blocka2a} 
& Yes 
& Theoretical  
& No \\
\midrule
AgentDID (This work) 
& Yes 
& \makecell{Experimental \\ \& Theoretical}
& Yes \\
\bottomrule
\end{tabular*}
\begin{tablenotes}
\footnotesize
\item * SSI denotes self-sovereign identity and CS denotes concurrent scalability. 
\end{tablenotes}
\end{table}

\section{Models and Preliminaries}  \label{sec:preliminaries}
In this section, we first introduce our model and then present the preliminaries that serve as the basic building blocks of our schemes. 

\subsection{Models}
\textbf{Participants.} This protocol involves four entities: the \textit{Controller}, the \textit{Issuer}, the \textit{Holder}, and the \textit{Verifier}. The \textit{Controller} initializes the identity of an AI agent by generating its cryptographic keys and registering a DID as the agent’s identity anchor. The \textit{Issuer} verifies identity-related claims associated with an AI agent and issues cryptographically endorsed VCs. The \textit{Holder} is the AI agent whose identity is to be authenticated.  It participates in the protocol by responding to authentication challenges and presenting its credentials. The \textit{Verifier} is the AI agent that initiates the authentication process. 
Before interaction, it evaluates the Holder’s identity and current operational state by issuing challenges and validating the responses and presented credentials. 

\textbf{Adversary.} We consider an adversary whose goal is to disrupt agent-to-agent (A2A) interactions by providing false information prior to interaction.  
The adversary may target AI agents involved in the system, including attempting to impersonate an agent or manipulating an agent to misreport its properties. 

First, the adversary may attempt \emph{identity forgery}. 
In this case, the adversary aims to create or present a fake AI agent that claims a legitimate identity, or to impersonate an existing \textit{Holder} when interacting with an honest \textit{Verifier}. 
Such attempts may involve fabricating identifiers or misusing valid credentials out of their intended context. 
Second, the adversary may attempt \emph{state forgery}.
Here, the adversary controls or compromises an AI agent and causes it to provide false or outdated runtime information, such as execution context, current workload, or claimed capabilities.
These falsified states are provided prior to A2A interaction and may mislead a \textit{Verifier} into engaging in unsafe or unsuitable interactions.

We assume that the \textit{Controller} and the \textit{Issuer} correctly perform identity initialization and credential issuance.
However, credentials issued by an \textit{Issuer} may be replayed or misused by an adversary, and runtime state information provided by AI agents cannot be assumed truthful without verification. 

\textbf{Security Goals.} Our goal under this model is to prevent identity forgery and state forgery, and to achieve correct A2A authentication and state verification. 
\begin{definition}[Security of AgentDID]  
Let $\mathcal{A}$ be any probabilistic polynomial-time (PPT) adversary that may corrupt an arbitrary subset of agents and the network. 
We say that AgentDID is \emph{secure} if, for any honest Verifier $V^*$, the probability that $\mathcal{A}$ causes $V^*$ to accept either an \emph{identity forgery} or a \emph{state forgery} is negligible in the security parameter $\lambda$.  
% \begin{itemize}
%   \item \textbf{Identity forgery:} $V^*$ accepts a purported Holder as some identity (DID) for which the adversary does not control the corresponding operational secret key and does not possess valid VCs bound to that DID.  
%   \item \textbf{State forgery:} $V^*$ accepts that a Holder satisfies a required dynamic state predicate, while the Holder is not actually in a state that satisfies this predicate at verification time.  
% \end{itemize}
\end{definition}

 In the rest of this section, we introduce the three key technologies adopted to achieve these design goals. 

\subsection{Decentralized Identifiers} 
DIDs~\cite{mazzocca2025survey} are W3C-standardized identifiers for digital identities. They allow entities to create and manage their identities autonomously, without relying on centralized identity providers.  

\begin{definition}[DID Operations]
The operations of a DID are formalized by a tuple of polynomial-time algorithms 
$\Pi_{\mathsf{DID}} = (\mathsf{KeyGen}, \mathsf{Create}, \mathsf{Resolve}, \mathsf{Update})$, defined as follows:
\begin{itemize}
    \item $(pk, sk) \leftarrow \mathsf{KeyGen}(1^{\lambda})$: Given a security parameter $\lambda$, the algorithm generates an asymmetric key pair consisting of a public key $pk$ and a private key $sk$. 

    \item $(\mathsf{did}, \mathsf{doc}) \leftarrow \mathsf{Create}(pk, sk)$:  
    Given a key pair $(pk, sk)$, the algorithm creates a globally unique DID string $\mathsf{did}$ and constructs an associated DID Document $\mathsf{doc}$ that binds the DID to the public key and other verification-related metadata.  

    \item $\mathsf{doc} \leftarrow \mathsf{Resolve}(\mathsf{did})$:  
    Given a DID string $\mathsf{did}$, the algorithm returns the latest corresponding DID Document. 

    \item $\mathsf{\{true, false\}} \leftarrow \mathsf{Update}(\mathsf{did}, \delta, sk)$:
    Given a DID $\mathsf{did}$, a state delta $\delta$, and a private key $sk$. The algorithm applies $\delta$ to update the current DID Document only if $sk$ corresponds to an authorized public key listed in the document. Returns $\mathsf{true}$ on success, $\mathsf{false}$ otherwise.
\end{itemize}
\end{definition} 

In our framework, the Controller adopts $\Pi_{\mathsf{DID}}$ to register the agent’s DID on-chain as its digital identity. This establishes a verifiable identity anchor that can be publicly resolved by any verifier, enabling decentralized and cross-domain authentication. 

\subsection{Verifiable Credentials} 
VCs are cryptographically signed statements issued by an issuer, attesting to specific claims about a holder. 

\begin{definition}[VC Operations]
The lifecycle of a VC is formalized by a tuple of polynomial-time algorithms $\Pi_{\mathsf{VC}} = (\mathsf{Request}, \mathsf{Issue}, \mathsf{Present}, \mathsf{Verify})$, defined as follows: 
\begin{itemize}
    \item $\mathsf{req} \leftarrow \mathsf{Request}(\mathsf{claims}, sk_H)$:  
    Given a set of claims $\mathsf{claims}$, the Holder uses its private key $sk_H$ to generate a signed credential request $\mathsf{req}$.

    \item $\mathsf{cred} \leftarrow \mathsf{Issue}(\mathsf{req}, sk_I)$:  
    The Issuer $I$ verifies the request and signs the attested claims using its private key $sk_I$, producing a VC $\mathsf{cred}$.

    \item $\mathsf{vp} \leftarrow \mathsf{Present}(\mathsf{cred}, \mathsf{nonce}, sk_H)$:  
    The Holder generates a VP $\mathsf{vp}$ by signing the credential $\mathsf{cred}$ together with a fresh challenge $\mathsf{nonce}$ issued by the Verifier to ensure that the presentation corresponds to the current verification request. 

    \item $\{\mathsf{true}, \mathsf{false}\} \leftarrow \mathsf{Verify}(\mathsf{cred}, pk_I)$:  
    Given the Issuer’s public key $pk_I$ and a credential $\mathsf{cred}$, the algorithm verifies the Issuer's digital signature and the integrity of the claims, outputting a boolean result.
\end{itemize}
\end{definition}

The set of claims carried by valid credentials constitutes verifiable evidence of a holder’s attributes and capabilities.  

\subsection{Publicly Detectable Watermark}

Publicly detectable watermarking (PDW) for Large Language Models (LLMs) aims to enable reliable attribution of generated text, while allowing watermark detection to be performed using public information. 

\begin{definition}[PDW Scheme]
PDW is formalized by a tuple of polynomial-time algorithms 
$\Pi_{\mathsf{PDW}} = (\mathsf{Setup}, \mathsf{Watermark}, \mathsf{Detect})$, defined as follows: 
\begin{itemize}
    \item $(k_{\mathsf{priv}}, k_{\mathsf{pub}}) \leftarrow \mathsf{Setup}(1^\lambda)$:  
    On input a security parameter $\lambda$, the setup algorithm outputs a private watermarking parameter $k_{\mathsf{priv}}$ and a corresponding public detection parameter $k_{\mathsf{pub}}$. 

    \item $t \leftarrow \mathsf{Watermark}(k_{\mathsf{priv}}, \rho)$:  
    Given the private watermarking parameter $k_{\mathsf{priv}}$ and a prompt $\rho$, the algorithm produces a text $t$ generated by the underlying language model with an embedded watermark. 

    \item $\{\mathsf{true}, \mathsf{false}\} \leftarrow \mathsf{Detect}(k_{\mathsf{pub}}, t^*)$:  
    Given the public detection parameter $k_{\mathsf{pub}}$ and a candidate text $t^*$, the detection algorithm determines whether the text contains a valid watermark. 
\end{itemize}
\end{definition}
The embedding procedure is kept private, while detection relies solely on public information, and recovering the embedding mechanism from the public detector is computationally infeasible. 

\section{AgentDID} \label{sec:AgentDID}
This section begins with an overview of the proposed framework, followed by a detailed description of each component. 

\subsection{Overview} 

\begin{figure}[h]
  \centering
  \includegraphics[width=\linewidth]{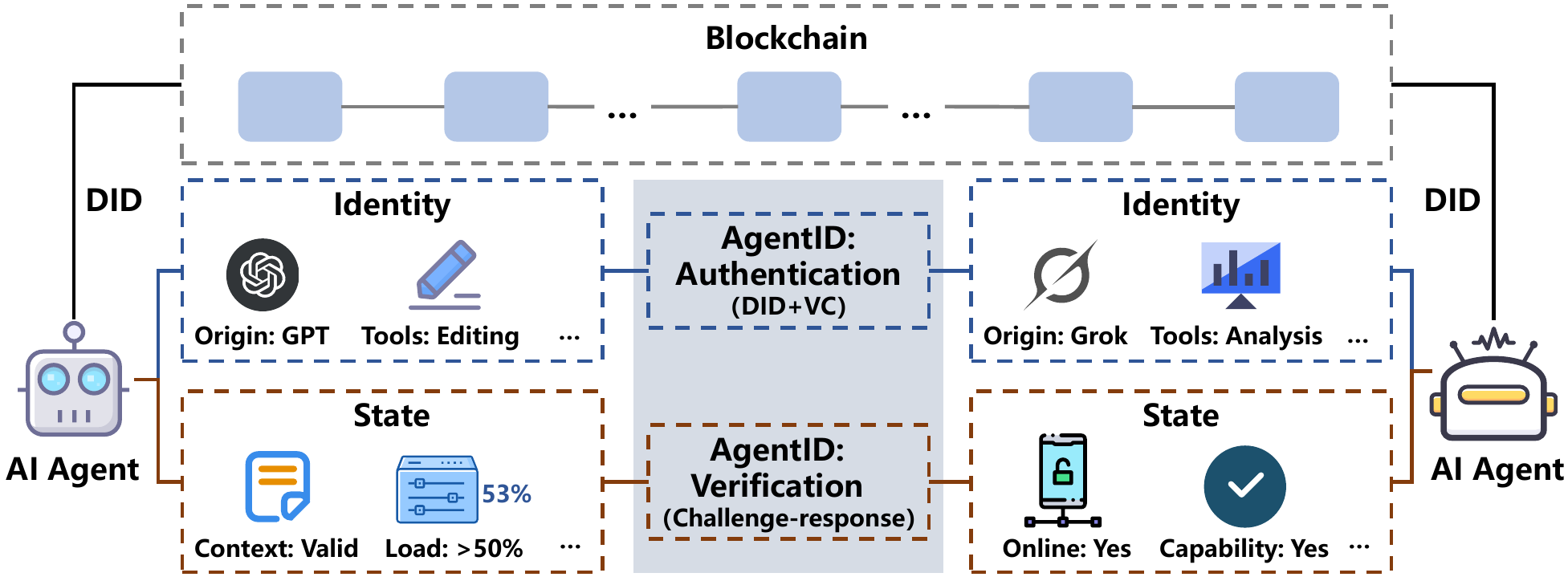}
  \caption{The overview of the framework.  In our design, an AI agent is described by a combination of static identity attributes and dynamic runtime states. Building on this abstraction, AgentDID enables identity authentication and state verification for A2A interactions. }
  \label{fig:overview}
\end{figure}

The framework represents an AI agent using two components: a static identity and a set of dynamic states. The static identity refers to relatively stable attributes of the agent that persist over extended periods, such as the underlying model and the tools it can access. The dynamic state captures runtime information that may change during execution, including context, workload, online status, and current capability availability. 

Each AI agent is associated with an on-chain DID, which serves as a stable identity reference. Based on this identifier, Issuers issue VCs that describe the agent’s long-term attributes, such as its underlying model and available tools. These credentials form the basis of the agent’s static identity and can be selectively disclosed in the form of VPs during authentication in different applications.   In addition to static attributes, the framework includes dynamic states that are verified at the time of interaction. When two agents interact, the requester acts as a verifier and initiates a challenge that targets specific runtime conditions, such as workload or current capability availability. The responding agent, as the Holder, must generate a proof to demonstrate that its current state satisfies the requested conditions. This process enables runtime state verification without relying on a central authority. 

Using this representation, AgentDID supports both identity authentication and dynamic state verification in A2A interactions. In the following subsections, we describe the details of identity generation, authentication, and dynamic state verification.    

\subsection{Identity Generation} 
The identity of an AI agent consists of a globally unique DID and a structured set of verifiable attributes. These attributes define the essential characteristics of the agent and are organized into three dimensions: provenance, capabilities, and compliance. The provenance dimension includes information such as the agent’s origin, the base model used, and training data sources. The capabilities dimension reflects the tools the agent can access and its supported functions. The compliance dimension refers to its alignment with regulatory requirements and audit records. To represent and verify these attributes in a decentralized and interoperable way, we adopt VCs. These credentials serve as cryptographic attestations issued by trusted issuers, allowing third parties to verify the agent’s identity and capabilities without requiring direct trust relationships. 

The complete identity generation process begins with the Controller registering a DID for the AI agent. The agent then submits signed claims for selected attributes to Issuers in order to request VCs. After evaluating the submitted claims, each Issuer issues the corresponding credential. Together, these credentials form a verifiable identity for the agent. 

\noindent \textbf{Register DID.} AgentDID adopts the W3C standard to create a DID for each AI agent.  
The Controller generates two independent asymmetric key pairs. The first is an administrative key pair, denoted as $(pk_\mathsf{admin}, sk_\mathsf{admin})$, which is used to perform identity management operations, such as key recovery and service endpoint updates. The second is an operational key pair, denoted as $(pk_\mathsf{op}, sk_\mathsf{op})$, which enables the AI agent to request VCs and sign VPs. Based on the administrative key, the Controller creates a DID and its initial DID Document using the procedure defined in $\Pi.\mathsf{Create}$. Subsequently, the Controller incorporates the operational key $pk_\mathsf{op}$ into the DID Document via $\Pi.\mathsf{Update}$. Consequently, the final DID Document includes both public keys. An example of the DID Document is shown in Fig.~\ref{lst:did-agent-doc}. Specifically, the \texttt{verificationMethod} field lists the cryptographic material for these keys using the Multibase format. To enforce privilege separation, $pk_{\mathsf{admin}}$ is referenced in the \texttt{capabilityInvocation} property, granting exclusive authority for document updates. Conversely, $pk_{\mathsf{op}}$ is associated with \texttt{authentication} and \texttt{assertionMethod}, empowering the agent to perform daily signing tasks without exposing the administrative root of trust. Additionally, the \texttt{service} field exposes a standardized communication endpoint for external interaction.  

\begin{figure}[!t]
    \centering
    \begin{minipage}{0.95\columnwidth}
    \begin{lstlisting}[linewidth=\columnwidth, language=json, firstnumber=1,
    keywords={"@context", "id", "type", "verificationMethod", "controller", "publicKeyMultibase", "capabilityInvocation", "authentication", "assertionMethod", "service", "type", "serviceEndpoint"}]
{
  "@context": [
    "https://www.w3.org/ns/did/v1",
    "https://w3id.org/security/suites/ed25519-2020/v1"
  ],
  "id": "did:agent:123456789abcdefghi",
  "verificationMethod": [
    {
      "id": "did:agent:123456789abcdefghi@@#admin-key@@",
      "type": "Ed25519VerificationKey2020",
      "controller": "did:agent:123456789abcdefghi",
      "publicKeyMultibase": "z6MkqREd9...AdminKeyRawData..."
    },
    {
      "id": "did:agent:123456789abcdefghi@@#op-key-1@@",
      "type": "Ed25519VerificationKey2020",
      "controller": "did:agent:123456789abcdefghi",
      "publicKeyMultibase": "z6MkwW8j3...OpKeyRawData..."
    }
  ],
  "capabilityInvocation": [
    "did:agent:123456789abcdefghi@@#admin-key@@"
  ],
  "authentication": [
    "did:agent:123456789abcdefghi@@#op-key-1@@"
  ],
  "assertionMethod": [
    "did:agent:123456789abcdefghi@@#op-key-1@@"
  ],
  "service": [
    {
      "id": "did:agent:123456789abcdefghi#agent-comm",
      "type": "AgentMessaging",
      "serviceEndpoint": "https://agent.example.com/api"
    }
  ]
}
    \end{lstlisting}
    \end{minipage}
    \caption{A DID Document for the AI agent.}
    \label{lst:did-agent-doc}
\end{figure}

\noindent \textbf{Request VCs.}  
After registering a DID, the AI agent must obtain VCs that confirm its core identity attributes. These credentials provide a reliable basis for subsequent authentication. The agent begins by submitting a signed request to an issuer. Each request includes a set of claims and is signed using the agent’s operational private key. Upon receiving the request, the Issuer verifies the signature and starts verifying each claim individually. 

The first step is to confirm the provenance of the agent. The Issuer must verify that the agent was created and managed by a valid Controller. This can be done by checking a trusted registry or by resolving a DID that is published on a website with verified domain ownership. The Issuer then requests a signed statement from the Controller to prove the relationship. This statement is verified using the Controller’s public key. 
Next, the Issuer verifies the core language model claimed by the AI agent using the previously defined PDW scheme. 
Specifically, the Issuer obtains the corresponding public detection parameter $k_{\mathsf{pub}}$ and challenges the agent with a fresh prompt $\rho$. 
The agent responds with a candidate text $t^{*}$, which is subsequently checked using $\Pi_\mathsf{PDW}.\mathsf{Detect}$.  
If the detection succeeds, the agent’s model claim is accepted. 

Then, the Issuer further evaluates whether the AI agent has access to the external tools it declares. 
This is performed by issuing test queries to each declared interface to confirm both tool availability and correct invocation behavior.
In addition, the Issuer assesses the agent’s functional capabilities by running standardized benchmark evaluations selected according to the agent’s declared domain.
These benchmarks may target different aspects of agent behavior, such as general planning ability or effective tool usage.  The evaluation outcomes are recorded in a structured format, including the benchmark name, version, numerical score, and a reference link to the complete test report, and are incorporated into the issued credential. 
As illustrated in Fig.~\ref{lst:vc-agent-static}, the capability assessment results are recorded in the \texttt{credentialSubject} field, including the applied benchmark (e.g., AgentBench), its version, the agent’s score, and a reference to the complete evaluation report.  

Finally, the Issuer verifies the agent’s claims about compliance. This task must be performed by issuers with the appropriate qualifications. For example, an agent that processes personal data must comply with relevant data protection rules. Once verified, the Issuer creates a structured claim and issues it as a credential. Each validated claim is included in a credential signed by the Issuer using the procedure defined in $\Pi_\mathsf{VC}.\mathsf{Issue}$. These credentials form the basis of a structured and verifiable agent identity.

\begin{figure}[htp]
    \centering
    \begin{minipage}{0.95\columnwidth}
    \begin{lstlisting}[linewidth=\columnwidth, language=json, firstnumber=1, keywords={"@context", "type", "name", "description", "credentialSubject", "id", "evaluation", "@type", "ratingSystem", "ratingVersion", "ratingValue", "bestRating", "dimensionScores", "alfworld_planning", "webshop_tool_use", "os_interaction", "lateral_thinking", "reportUrl", "datasetHash"}]
{
  "@context": [
    "https://www.w3.org/ns/credentials/v2",
    "https://schema.org"
  ],
  "type": ["VerifiableCredential", "AgentCapabilityCredential"],
  "name": "Agent Capability Assessment",
  "description": "Verified performance metrics evaluating agent planning and tool usage capabilities.",
  "credentialSubject": {
    "id": "did:agent:123456789abcdefghi",
    "evaluation": {
      "@type": "Rating",
      "ratingSystem": "AgentBench v0.2 (Comprehensive)",
      "ratingVersion": "v0.2.1",
      "ratingValue": "0.785",
      "bestRating": "1.000",
      "dimensionScores": {
        "alfworld_planning": 0.82,
        "webshop_tool_use": 0.75,
        "os_interaction": 0.68,
        "lateral_thinking": 0.89
      },
      "reportUrl": "https://example.eval.org/reports/agent-bench/uuid-550e8400-e29b",
      "datasetHash": "sha256:e3b0c44..."
    }
  }
}
    \end{lstlisting}
    \end{minipage}
    \caption{Example VC for static agent attributes.}
    \label{lst:vc-agent-static}
\end{figure}

\subsection{Identity Authentication and State Verification} \label{sec:agent_identity_authentication}
Building on the established digital identity of the AI agent, this section presents the interaction verification process. 
As shown in Fig.~\ref{fig:identity_authentication_state_verification}, the process consists of two phases: identity authentication and state verification. 

The following subsections describe these two phases in detail. 

\begin{figure}[htp]
  \centering
  \includegraphics[width=0.98\linewidth]{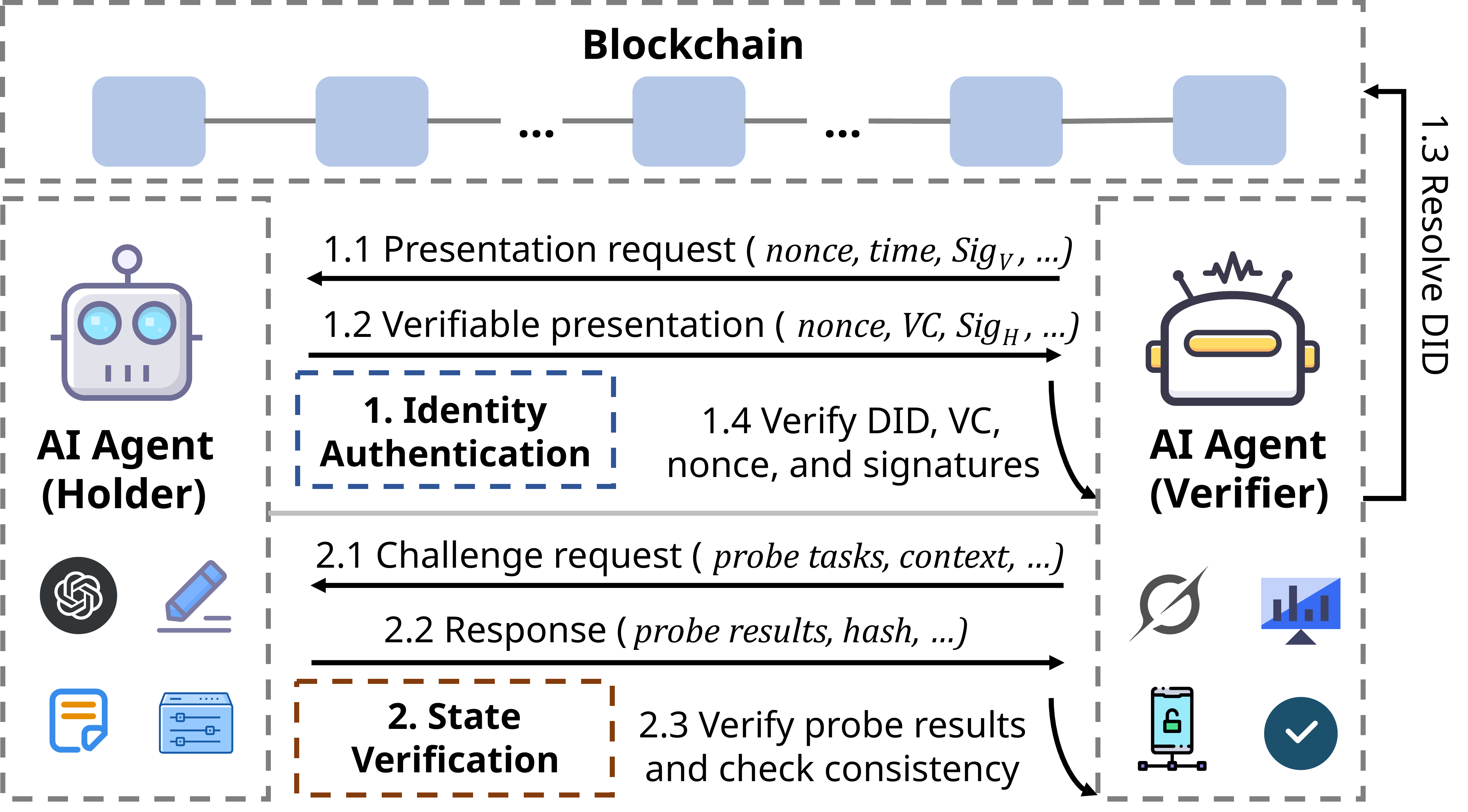}
  \caption{The workflow for identity authentication and state verification.}
  \label{fig:identity_authentication_state_verification}
\end{figure}

\noindent \textbf{Identity Authentication.} This process constitutes one agent verifying the static identity of another agent, and AgentDID adopts DIDs and VCs to achieve this purpose. 

Consider two AI agents, a Holder and a Verifier, engaged in an identity authentication protocol. The authentication process is initiated by the Verifier, who generates a one-time random value (nonce) and transmits it to the Holder. The nonce ensures the freshness of the Holder's subsequent response, thereby mitigating replay attacks. Upon receiving the challenge, the Holder constructs a VP. The VP incorporates the received nonce together with the VCs required for the interaction. In particular, the Holder invokes $\Pi_{\mathsf{VC}}.\mathsf{Present}$ to generate the VP and subsequently signs the resulting presentation using its operational private key. 

After receiving the VP, the Verifier performs a multi-step verification procedure. First, the Verifier resolves the Holder's DID to retrieve the operational public key contained within the corresponding DID Document. This public key is then used to verify the VP's digital signature. A successful signature verification confirms that the VP was generated by the Holder. The Verifier also checks that the nonce embedded in the VP matches the previously issued challenge, ensuring that the VP constitutes a timely response rather than a replayed message.  Next, the Verifier assesses the validity of the enclosed VCs using $\Pi_{\mathsf{VC}}.\mathsf{Verify}$.
This evaluation involves confirming that the DIDs of all credential issuers appear in the Verifier's trusted issuer list and verifying the integrity of each VC through digital signature validation. The Verifier must further ensure that the subject of each VC corresponds to the Holder, preventing the misuse of stolen  credentials. Finally, the Verifier checks the validity period of the VCs to ensure that none have expired. 

\noindent \textbf{State Verification.} This process corresponds to one agent verifying the dynamic state of another agent, and AgentDID designs state-specific challenge tasks to validate the correctness of the peer’s current state. 

VCs establish the identity of an AI agent, but they are essentially static: their contents remain unchanged over long periods and do not capture the agent's real-time operational state during interactions. In contrast, the state of an AI agent evolves continuously, and factors such as its current readiness, operational status, and task alignment may affect its ability to execute a given task. Consequently, static VCs alone cannot verify these dynamic properties.
To address this limitation, we introduce a state verification mechanism that evaluates the AI agent's dynamic state. State verification consists of two components: the \emph{Readiness Probe}, which determines whether the agent is in a state suitable for executing the task, and the \emph{Context Consistency Check}, which verifies that the holder maintains a synchronized interaction context with the verifier, ensuring that no context has been lost due to internal faults or resets. These two components provide dynamic verification that complements the static nature of VCs.  

\noindent \textit{\underline{Readiness Probe.}}
This process evaluates whether the agent can complete a probe task. A successful response provides the Verifier with several forms of behavioral evidence. First, successful message exchange indicates that the agent's service process is online and reachable. Second, correct interpretation and execution of the probe instructions show that the core inference engine is operational. Third, the probe task requires invoking one or more specified tools, and successful tool invocation confirms that the execution path from instruction parsing to tool usage is functioning as expected. 
For cost-sensitive tasks, the Verifier may additionally estimate token usage during the probe to infer the expected cost of subsequent interactions. 
Using a single probe, the Verifier can also measure response latency to estimate the agent's current load and decide whether to proceed with the interaction. If the Holder is unavailable or exhibits high latency, the Verifier may wait or select an alternative holder. 

The probe task is standardized at the protocol layer, defining a uniform interface and data structure that all parties follow. The task content itself, however, is derived from the specific interaction to be performed. The Verifier dynamically populates the probe with task-relevant parameters, such as the critical tools that need to be invoked. 
Fig.~\ref{lst:probe_task} provides an example of a comprehensive probe task template that combines text processing with multiple tool invocations. The \texttt{template\_id} field provides a unique identifier for this task structure. The \texttt{template\_str} defines a structured prompt with placeholders like \texttt{{{input\_text}}}, separating the fixed instruction logic from the dynamic input data, which is generated by the Verifier for each probe to prevent replay attacks. The \texttt{required\_tool\_names} field explicitly lists the tools that must be used, in this case for querying the current date and performing a SHA-256 hash calculation, thereby enabling cryptographic verification of the tool invocation path. 
% Finally, the \texttt{timeout\_ms} field is set to a descriptive placeholder, \texttt{Dynamically Calculated Latency}, indicating that the response deadline is not fixed but is computed by the verifier based on factors such as network conditions and task complexity, providing a realistic and context-aware SLA for the Holder agent.
Finally, the \texttt{timeout\_ms} field uses a descriptive placeholder, \texttt{Dynamically Calculated Latency}. 
This indicates that the response deadline is not fixed in advance, but instead computed by the Verifier at runtime based on factors such as network conditions and task complexity. 
%As a result, the verifier can enforce a more realistic and context-aware SLA for the Holder agent. 

\begin{figure}[htp]
    \centering
    \begin{minipage}{0.95\columnwidth}
    \begin{lstlisting}[linewidth=\columnwidth, language=json, firstnumber=1, keywords={"template_id", "description", "template_str", "required_tool_names", "timeout_ms"}]
{
    "template_id": "tpl_comprehensive_check",
    "description":"Comprehensive Check: Summarizes text, queries the current time, and hashes the original input.",
    "template_str": "Please perform three actions: 1. Summarize the text: '{{input_text}}'. 2. Get the current UTC date using '{{required_tools[0]}}'. 3. Calculate the SHA-256 hash of the original input text using '{{required_tools[1]}}'. Respond in a JSON object with keys 'summary', 'current_date', and 'text_hash'.",
    "required_tool_names": ["get_current_utc_date", "get_hash"],
    "timeout_ms": "Dynamically Calculated Latency"
  }
    \end{lstlisting}
    \end{minipage}
    \caption{Example probe task for availability verification.}
    \label{lst:probe_task}
\end{figure}

\noindent \textit{\underline{Context Consistency Check.}}
This process ensures that the AI agent has not lost context due to internal factors such as stateless instance rotation, cache resets, or internal logic faults during an interaction, or before beginning a new interaction that depends on previously established context. 

Context consistency is verified by checking the hash of the complete serialized context. When contextual alignment between both parties is required, such as in multi-turn or long-duration tasks, the Verifier issues a consistency check request to the Holder. Simultaneously,the Verifier computes $h_{\text{verifier}}$, the hash of its local context sequence prior to this request. Upon receiving the request, the Holder computes the hash of its context sequence excluding the received request to obtain $h_{\text{holder}}$, signs it, and sends it back to the Verifier. After receiving the response, the Verifier validates the signature using the Holder's public key to ensure the authenticity of the sender and the integrity of the transmitted data. Subsequently, the Verifier compares $h_{\text{verifier}}$ with $h_{\text{holder}}$ to determine whether the contexts held by both parties are consistent.  

In summary, this protocol employs cryptographic mechanisms to authenticate the identity of the AI agent. Based on this trust foundation, the protocol further verifies the agent's dynamic state, including its readiness and contextual alignment at the time of interaction. These mechanisms together support the construction of a secure multi-agent ecosystem. 

\section{Security Analysis} \label{sec:security-analysis}
In this section, we provide the security analysis of AgentDID. 

\begin{theorem}
Assume that (i) the signature schemes used in $\Pi_{\mathsf{DID}}$ and $\Pi_{\mathsf{VC}}$ are existentially unforgeable under chosen-message attacks (EUF-CMA), (ii) the VC system $\Pi_{\mathsf{VC}}$ is sound, (iii) the publicly detectable watermarking scheme $\Pi_{\mathsf{PDW}}$ is secure, and (iv) the hash function used in the context consistency check is collision-resistant. 
Then, for any PPT adversary $\mathcal{A}$ that may corrupt an arbitrary subset of agents and the network, and for any honest Verifier $V^*$, the probability that $\mathcal{A}$ causes $V^*$ to accept either an identity forgery or a state forgery is negligible in the security parameter $\lambda$.  
\end{theorem}
\begin{IEEEproof} 
    Let $\mathcal{A}$ be any PPT adversary interacting with AgentDID and an honest Verifier $V^*$. 
We consider the two bad events from the definition: identity forgery and state forgery. 
We show that each such event can only occur with negligible probability unless one of the assumptions is broken. 

{\bf Identity Forgery.}
Recall that an AI agent’s identity is anchored in a DID created via $\Pi_{\mathsf{DID}}.\mathsf{Create}$ and bound to an operational public key $pk_{\mathsf{op}}$ in the corresponding DID Document.
During authentication, the Holder constructs a VP using $\Pi_{\mathsf{VC}}.\mathsf{Present}$, embeds a fresh nonce generated by $V^*$, and signs the resulting presentation with the associated secret key $sk_{\mathsf{op}}$.
The honest Verifier $V^*$ then resolves the Holder’s DID using $\Pi_{\mathsf{DID}}.\mathsf{Resolve}$ to obtain $pk_{\mathsf{op}}$, verifies the signature on the VP under $pk_{\mathsf{op}}$, checks the embedded nonce for freshness, and verifies each enclosed credential via $\Pi_{\mathsf{VC}}.\mathsf{Verify}$ while enforcing issuer trust lists, subject binding, and validity periods.

Suppose $\mathcal{A}$ succeeds in an identity forgery.
Such an event can only arise in one of the following ways.
First, $\mathcal{A}$ may create a VP that passes signature verification under $pk_{\mathsf{op}}$ and includes the fresh nonce of $V^*$ without knowing $sk_{\mathsf{op}}$, which contradicts EUF-CMA security of the signature scheme.
Second, $\mathcal{A}$ may produce a VC that passes $\Pi_{\mathsf{VC}}.\mathsf{Verify}$ under a trusted Issuer’s public key but was never issued for that Holder and DID, which contradicts the soundness of $\Pi_{\mathsf{VC}}$.
Third, $\mathcal{A}$ may change the binding between a DID and its public key so that $\Pi_{\mathsf{DID}}.\mathsf{Resolve}$ returns an attacker-favorable key, which contradicts the correctness of the DID infrastructure assumed by the system.
%
% In addition, confirming the underlying model is treated as part of the agent’s static identity.
% Model provenance is validated during credential issuance by means of the publicly detectable watermarking scheme $\Pi_{\mathsf{PDW}}$.
% Any attempt to falsely claim a base model would require producing texts that are accepted by $\Pi_{\mathsf{PDW}}.\mathsf{Detect}$ for a model the adversary does not actually possess, which contradicts the assumed security of $\Pi_{\mathsf{PDW}}$.
%Therefore, misrepresenting the base model is also a special case of identity forgery.

In all cases, a successful identity forgery would allow us to construct a reduction that breaks one of the underlying assumptions with non-negligible probability.
Hence, the success probability of identity forgery is negligible in $\lambda$.
 
{\bf State Forgery.}
AgentDID verifies dynamic state along two dimensions: readiness (via probe tasks) and context consistency (via hashed context commitments). 
Unlike static attributes such as the underlying model or provenance, which are established during credential issuance, both readiness and the ability to exercise the required capabilities at the moment of interaction are treated as part of the Holder’s dynamic state.

\textit{Readiness.}
The readiness probe requires the Holder to complete a standardized probe task that combines instruction-following, tool invocation, and a timeout constraint. 
Each probe instance includes fresh, Verifier-chosen inputs and is bound to the Holder through signatures under $sk_{\mathsf{op}}$. 
An honest Verifier accepts the readiness predicate only if the probe response is timely and consistent with the specified behavior. 
If $\mathcal{A}$ could cause $V^*$ to accept a readiness predicate when the Holder is not actually in a state that satisfies it, then $\mathcal{A}$ must either forge a valid signed response on behalf of the Holder without executing the required actions, or simulate the expected tool behavior and protocol traces for capabilities that it does not possess, while still satisfying all syntactic and semantic checks. 
The first case contradicts signature unforgeability. 
The second case contradicts the threat model and the fact that readiness predicates are defined exactly in terms of externally observable behaviors. 
Therefore, the probability of a successful readiness-related state forgery is negligible.

\textit{Context consistency.}
For context consistency, both parties compute a hash of their current serialized context; the Holder sends $h_{\mathsf{holder}}$ signed under $sk_{\mathsf{op}}$, and the Verifier computes $h_{\mathsf{verifier}}$ over its own context and compares the two values. 
The Verifier accepts the context predicate only if the signature is valid and $h_{\mathsf{holder}} = h_{\mathsf{verifier}}$. 
A successful context-based state forgery would require that $V^*$ accepts while the two internal contexts differ. 
This can only happen if either $\mathcal{A}$ forges a signature on a hash value not computed from the Holder’s actual context, or finds two distinct contexts that hash to the same value. 
The first case violates EUF-CMA security of the signature scheme, and the second breaks the collision resistance of the hash function. 
Both are assumed to be infeasible for PPT adversaries, so the probability of such a forgery is negligible. 

%Finally, note that the effective ability to exercise the claimed tools or capabilities at interaction time forms part of the Holder’s dynamic state, since these capabilities are validated through the probe during the interaction rather than through static issuance. 
%Consequently, misrepresenting such runtime capabilities without being able to complete the probe is exactly a form of state forgery and ruled out by the above argument.  
% In all cases, a successful state forgery would allow us to construct a reduction that breaks one of the underlying assumptions with non negligible probability. 
% Hence, the success probability of state forgery is negligible in $\lambda$. 

% \paragraph{Model and capability misrepresentation.}
% In the identity generation phase, issuers verify model provenance and capabilities before issuing VCs. 
% Model claims are checked via the publicly detectable watermarking scheme $\Pi_{\mathsf{PDW}}$: the issuer uses a public detection key $k_{\mathsf{pub}}$ and fresh prompts to test the agent’s responses. 
% A successful misrepresentation (claiming an unowned model) would give an adversary that produces texts accepted by $\Pi_{\mathsf{PDW}}.\mathsf{Detect}$ for a model it does not actually use, contradicting the assumed security of $\Pi_{\mathsf{PDW}}$. 
% Similarly, misrepresenting tools or benchmarked capabilities while still obtaining valid VCs would break the soundness of the VC issuance and verification process.  

In summary, we have shown that any PPT adversary that achieves identity forgery or state forgery with non negligible probability can be used to construct a PPT algorithm that breaks one of the underlying assumptions. 
By contradiction, the probability that $\mathcal{A}$ causes an honest Verifier $V^*$ to accept either an identity forgery or a state forgery is negligible in $\lambda$.  
Hence, AgentDID is secure according to the definition. 
\end{IEEEproof}

\section{Implementation and Evaluation}  \label{sec:evaluation}
In this section, we describe the implementation of AgentDID and evaluate its performance through experiments. Our code repository is open sourced on GitHub\footnote{https://anonymous.4open.science/r/AgentDID/}. 

\subsection{Experimental Setup}
% Experiments are conducted on a virtual machine equipped with a 24-core Intel(R) Xeon(R) Silver 4214R CPU (providing 48 logical cores) and 64 GB of RAM, running the Ubuntu 22.04.1 LTS operating system. The prototype system is primarily developed using Python 3.11.14 to implement DID lifecycle management, the VC issuance process, and the identity authentication protocol. To ensure strict compatibility with W3C standards for DID resolution, we adopted a hybrid architecture: the core logic utilizes the \texttt{Web3.py} library for blockchain interactions and the \texttt{eth-account} library for ECDSA cryptographic operations, while the DID resolution module invokes Node.js (v18.20.8) scripts via Python subprocesses to leverage the official \texttt{ethr-did-resolver} library. The system integrates the LangChain and LangGraph frameworks to orchestrate the AI agent’s reasoning logic, employing the Qwen-Turbo and Qwen-Plus large models (via the DashScope API) as the underlying intelligence engines. On-chain experiments are deployed on the Ethereum Sepolia Testnet, interacting with the canonical ERC-1056 contract. 

Experiments are conducted on a virtual machine equipped with a 24-core Intel(R) Xeon(R) Silver 4214R CPU (48 logical cores) and 64\,GB of memory, running Ubuntu 22.04.1 LTS. 
The prototype system is primarily implemented in Python~3.11.14 and supports DID lifecycle management, verifiable credential issuance, and identity authentication workflows. 
The overall design conforms to the W3C DID specifications.  

The system adopts a hybrid implementation architecture. 
Core components are implemented in Python, leveraging the \texttt{Web3.py} library for blockchain interactions and the \texttt{eth-account} library for ECDSA-based cryptographic operations. 
For DID resolution, the system invokes Node.js (v18.20.8) modules via subprocess calls to integrate the official \texttt{ethr-did-resolver} implementation. For agent execution and coordination, the system integrates the LangChain and LangGraph frameworks to orchestrate reasoning workflows, and employs the Qwen-Turbo and Qwen-Plus large language models through the DashScope API as the underlying inference engines. 
All on-chain experiments are conducted on the Ethereum Sepolia testnet, interacting with the standard ERC-1056 registry contract\footnote{Address: \texttt{0x03d5003bf0e79C5F5223588F347ebA39AfbC3818}} for decentralized identifier management. 

Specifically, we describe the implementation details of several key design components as follows. 

\textbf{Identity Generation.}
The system incorporates a local key management component for cryptographic identity generation and protection. Specifically, we employ the \texttt{eth-account} library to generate ECDSA key pairs over the  \texttt{secp256k1} curve, which are securely persisted in a local JSON-based keystore. 
To enable VC issuance, we implement a lightweight Issuer service using the Flask framework. The Issuer maintains a set of standardized VC templates defined using JSON Schema and stored locally. Upon receiving a credential issuance request, the service instantiates the corresponding template by binding the applicant’s DID and produces a cryptographic proof using the Issuer’s private key. This process ensures that all issued credentials conform to the W3C VC data model and can be independently verified by external parties. 

{\bf Identity Authentication.}
To leverage the official \texttt{ethr-did-resolver}, the system performs DID resolution by invoking Node.js from the Python runtime. The associated overhead is mitigated through an in-memory caching layer, which reduces repeated invocations and redundant blockchain I/O. 
For P2P interactions, each agent exposes a lightweight HTTP endpoint to handle authentication requests. During verification, the runtime validates the VP by recovering the signing address from the VP proof using \texttt{eth-account} and matching it against the operational keys specified in the resolved DID Document. Additional checks enforce timestamp freshness and nonce consistency. 
The system further validates embedded VCs by recovering the signing addresses from their respective proofs and verifying them against the Issuer’s authorized keys. 

{\bf State Verification.} 
For the Readiness Probe, we use LangChain to orchestrate the agent’s reasoning process. Local tools (e.g., hash calculators) are exposed through Python decorators and invoked during execution. To validate the task outcome, the Verifier applies an LLM-as-a-Judge approach by running a parallel inference chain that checks whether the Holder’s natural language response aligns with the concrete outputs produced by the tools. This semantic check is complemented by the deterministic verification of the hash values. 
For the Context Consistency Check, the system enforces a canonical serialization of interaction history. We implement a canonicalization function that recursively orders JSON keys and removes whitespace before computing a SHA-256 hash. This guarantees that the same session state yields an identical digest across different runtime environments. The Verifier detects any modification of the interaction history by comparing the locally computed hash with the signature returned by the Holder.

\subsection{Performance Evaluation}
In this section, we present the experimental performance and provide a detailed analysis. 

{\bf Identify Generation.}
\begin{table}[b]
  \centering
  \caption{On-Chain Gas Consumption, Financial Costs, and Latency for Identity Initialization}
  \label{tab:onchain_costs}
  \begin{tabular*}{\columnwidth}{@{\extracolsep{\fill}} l c c c }
    \toprule
    \textbf{Operation} & \textbf{Gas Used} & \textbf{Cost (USD)} & \textbf{Latency (s)} \\
    \midrule
    DID Registration & 58,238 & \$0.88 & 15.37 \\
    \bottomrule
  \end{tabular*}
\end{table}
We evaluate the initialization cost and storage overhead associated with agent identity. The evaluation consists of 100 serial rounds of DID registration, including on-chain anchoring via implicit anchoring and the addition of operational keys through smart contract calls. For each round, we record the gas consumption and transaction latency of the corresponding on-chain operations. In addition, we measure the average storage overhead per VC.

Table~\ref{tab:onchain_costs} reports the on-chain cost of DID registration. The results show that on-chain write operations introduce a confirmation latency on the order of seconds (15.37 s), which originates from blockchain confirmation and occurs during the initialization phase. The experiments run on the Sepolia testnet. Financial cost estimates are computed using Ethereum Mainnet averages over the period from December 1, 2024 to November 30, 2025, based on an average gas price of 4.88 Gwei obtained from Etherscan\footnote{\url{https://etherscan.io/chart/gasprice}} 
 and an average ETH price of \$3,121.34 from CoinGecko\footnote{\url{https://www.coingecko.com/}} 
. Under these parameters, the estimated monetary cost per DID registration is below one dollar. In addition, the average storage overhead per VC is 1.23 KB. 

{\bf Concurrency Evaluation.} 
%\paragraph{Concurrency and Scalability Evaluation.} 
\begin{figure}[htp]
  \centering
\includegraphics[width=0.46\textwidth]{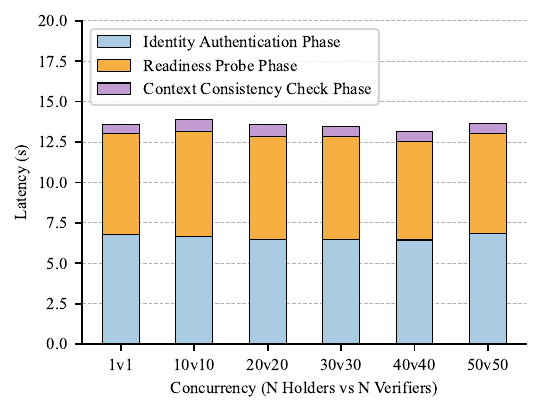}
  \caption{Latency breakdown across protocol phases under increasing concurrency. The x-axis denotes the number of concurrently authenticated agent pairs, and the y-axis reports the corresponding latency. 
} 
  \label{fig:latency_breakdown}
\end{figure}
\begin{figure}[t]
  \centering
  \includegraphics[width=0.46\textwidth]{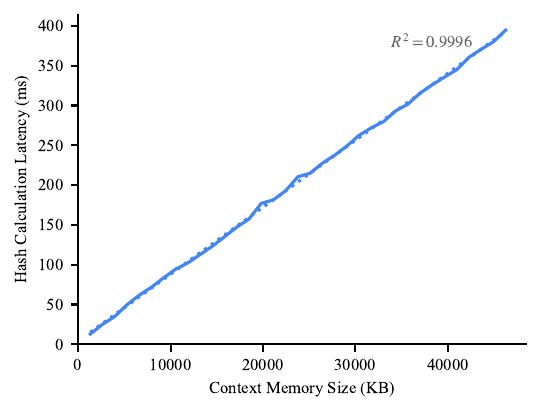}
  \caption{Latency of context hash computation under varying context sizes. The hash computation latency scales linearly with the size of the context memory ($R^2 = 0.9996$), remaining at the millisecond level even for large contexts, indicating that context consistency verification introduces negligible overhead.}
  \label{fig:context_consistency_verification}
\end{figure} 
\begin{figure}[b]
  \centering
  \includegraphics[width=0.46\textwidth]{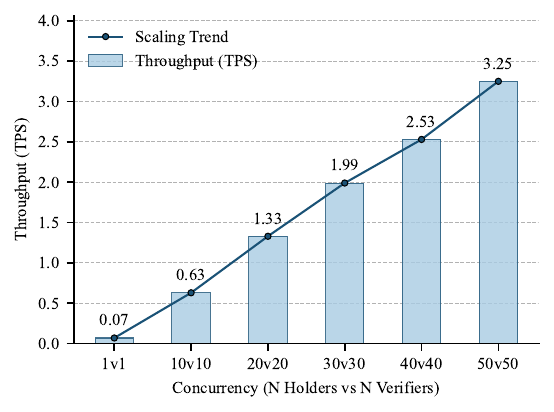}
  \caption{System throughput under increasing concurrency. The x-axis denotes the number of concurrent Holder--Verifier pairs, and the y-axis reports the system throughput. } 
  \label{fig:throughput_trend}
\end{figure} 
To evaluate the execution efficiency and scalability of the protocol under increasing agent populations, we conduct a concurrency experiment covering the complete protocol workflow. The experiment simulates a setting in which $N$ Holders interact with $N$ Verifiers concurrently, where $N \in \{1, 10, 20, 30, 40, 50\}$.  
Fig.~\ref{fig:latency_breakdown} reports the latency breakdown across protocol phases under increasing concurrency. %, showing that different phases exhibit distinct latency characteristics. 

For the Identity Authentication phase, the observed latency is approximately 6.5\,s, as indicated by the blue bars in Fig.~\ref{fig:latency_breakdown}. This latency is primarily incurred by blockchain-related operations, including on-chain DID resolution and Issuer DID resolution during VP verification. These steps require issuing queries to the blockchain and waiting for the corresponding responses and confirmations. As a result, the majority of the elapsed time is spent waiting for blockchain responses and network communication, while the time consumed by local protocol execution, such as cryptographic verification, constitutes only a small portion of the overall latency.  

For the Readiness Probe, we simulate a representative agent task in which the agent generates a one-sentence summary of a randomly selected text snippet, while concurrently invoking local tools to compute the hash of the text and retrieve the current system time. As indicated by the yellow bars in Fig.~\ref{fig:latency_breakdown}, the latency of this phase is comparable to that of the Identity Authentication phase. This latency is primarily incurred by waiting for responses from the external LLM API, as well as by the execution of tool-related logic during the task, rather than by blockchain interactions or local cryptographic processing.

For the Context Consistency Check, the latency contribution remains below one second across all concurrency levels, as shown by the purple bars in Fig.~\ref{fig:latency_breakdown}. This phase involves canonicalizing the session history and computing a cryptographic hash over the resulting representation.  To further quantify the cost of hashing large contexts, we conduct a microbenchmark that measures the time required to compute cryptographic hashes over increasing context sizes. The benchmark evaluates contexts of up to 10 million tokens (approximately 40\,MB), corresponding to the context window size supported by current long-context models. As shown in Fig.~\ref{fig:context_consistency_verification}, the hash computation latency increases linearly with the context size ($R^2 = 0.9996$). For contexts of approximately 40\,MB, the measured latency remains on the order of milliseconds.  

Overall, as the concurrency level increases from 1v1 to 50v50, the total protocol latency remains stable at approximately 13.5\,s.

Fig.~\ref{fig:throughput_trend} shows the system throughput under increasing concurrency. As the concurrency level scales from 1v1 to 50v50, the throughput rises from 0.07~TPS to 3.25~TPS and exhibits an approximately linear relationship with the number of concurrent Holder--Verifier pairs. This result shows that higher concurrency leads to a proportional increase in completed protocol executions per unit time.

\section{Conclusion}
\label{sec:conclusion}
This work studies identity authentication and state verification for AI agents in decentralized environments and proposes AgentDID as a concrete design that addresses these requirements. The framework shows how DIDs and VCs can be used to support self-managed agent identities and cross-system authentication without centralized control. By introducing a challenge–response mechanism, AgentDID further enables verification of execution-time conditions that are not captured by static credentials. The implementation and throughput experiments indicate that the proposed design supports concurrent identity authentication and state verification under workloads associated with large numbers of AI agents. 

In the future, we will focus on privacy considerations in AI agent interactions. During agent-to-agent and agent-to-service interactions, identity authentication and state verification may involve the exchange of information related to execution context, task progress, or capability availability. When such interactions occur repeatedly or across multiple parties, this information can reveal sensitive aspects of agent behavior and operational patterns. We will investigate how privacy can be preserved during agent interactions while still enabling correct identity authentication and state verification. 

\section*{Acknowledgment}
This study was supported by the National Natural Science Foundation of China (No. U25A20425, 62302266, 62232010, U23A20302, U24A20244), and the Research Project of Quancheng Laboratory, China under Grant No. QCL20250106.

\bibliographystyle{IEEEtran}
\bibliography{references}

\end{document}